\documentclass[conference,a4paper]{IEEEtran}
\usepackage{epsfig}
\usepackage{times}
\usepackage{float}
\usepackage{afterpage}
\usepackage{amsmath}
\usepackage{amstext}
\usepackage{amssymb,bm}
\usepackage{latexsym}
\usepackage{color}
\usepackage{graphicx}
\usepackage{amsmath}
\usepackage{amsthm}
\usepackage{graphicx}
\usepackage[center]{caption}
\usepackage{pstricks}
\usepackage{booktabs}
\usepackage{multicol}
\usepackage{lipsum}
\usepackage{dblfloatfix}
\usepackage{fixltx2e}
\usepackage{algorithm}
\usepackage[noend]{algpseudocode}
\usepackage{subcaption}
\usepackage{tikz}
\usepackage{thmtools,thm-restate}
\usepackage[normalem]{ulem}

\makeatletter
\def\BState{\State\hskip-\ALG@thistlm}
\makeatother

\usepackage{multirow}
\usepackage{tabularx}
\usepackage{pbox}
\usepackage{rotating}

\allowdisplaybreaks
\newtheorem{thm}{Theorem}

\newtheorem{lem}[thm]{Lemma}

\begin{document}
	
	\title{A Distortion Based Approach \\ for Protecting Inferences}
	\author{
		\IEEEauthorblockN{Chi-Yo Tsai, Gaurav Kumar Agarwal, Christina Fragouli, Suhas Diggavi }
		Department of Electrical Engineering, UCLA, Los Angeles, USA\\
		Email: \{chiyotsai, gauravagarwal, christina.fragouli, suhasdiggavi\}@ucla.edu
		\thanks{This work was supported in part by NSF grants 1321120 and 1514531.
		}

	}
	\IEEEoverridecommandlockouts
	
	\maketitle
	
	\begin{abstract}
		Eavesdropping attacks in inference systems aim to learn not the raw
		data, but the system inferences to predict and manipulate
		system actions.  We argue that conventional {information security } measures can be ambiguous on the adversary's estimation abilities, and adopt instead a distortion based framework {that enables to operate over a metric space}. We show that requiring perfect distortion-based security is more frugal than requiring perfect
		{information-theoretic} secrecy even for block length one codes, offering in some
		cases unbounded gains. Within this framework, we design algorithms
		that enable to efficiently use shared randomness, and show that each
		{bit of} shared random key is exponentially useful in security.
	\end{abstract}

	\section{Introduction}
	The operation of {Cyber Physical Systems (CPS), such as transportation systems and the electrical grid,}
	increasingly relies on networked data collection and inference systems.
	Such systems distributively collect and compute functions of data for control decisions. Adversarial attacks in such systems aim to learn not the raw data but the system inferences to predict and manipulate the system actions.    
	Most of the established security theory is tailored on protecting data;
	in contrast, in this paper,  we discuss how to  protect inference processes and decisions.
	
	
	Our first observation is that {information security measures
		(such as the conventional equivocation rate)}
	can be ambiguous on how much an eavesdropper has learned {about a CPS system operation}.
	Assume that a control function takes values uniformly at random in $\{1,2,\ldots, 20\}$, and the associated action is proportional to the function value. The entropy of this function is $\log(20)$. Now consider two cases. For case I, at the end of a transmission, the eavesdropper knows that the function value is in $\{1, 2\}$, each with probability $1/2$. For case II, the eavesdropper knows that the function value is in $\{1, 20\}$, each with probability $1/2$. Both cases have an equivocation rate of $1$ (i.e. the entropy of the function given the information the eavesdropper has, is $1$). But in the first case, the two possibilities are much closer than the second case, so the eavesdropper can predict the system action with high precision. Similarly, if an eavesdropper knows that the function value is uniformly at random in  $\{1, 20 \}$, this would still be less useful for the eavesdropper than knowing it is in $\{1, 2, 3, 4\}$, although the equivocation rate is lower. 
	
	{In this paper, we use instead a distortion measure for security that aims to maximize the square error difference between the eavesdropper's estimate and the true value of a function.
		Squared error distortion is a  widely accepted metric in estimation and control, and can capture how much an adversary has learned about core attributes of a CPS system, such as the system state. 
	}
	
	%
	
	
	
	
	Moreover, while for perfect information theoretic security we require shared keys of size equal to the message entropy,
	requiring perfect distortion-based security is much more frugal.
	For example, if we again take the example of a control function taking values uniformly at random from $\{1,2,\ldots, 20\}$, we show in Theorem~\ref{thm:optimality} that only one bit of shared key is sufficient to guarantee  perfect distortion based security, as opposed to the $\log(20)$ bits of shared keys for information theoretic security. We obtain this result using a scalar coding scheme of block length one.
	
	We design schemes to achieve distortion based security.
	Our main contributions are:\\
	$\bullet$ We show that for a single source, each additional {bit of} shared key is exponentially useful in distortion based security.\\
	$\bullet$ We design a polynomial time algorithm and show it is optimal for regularly spaced function values.\\
	$\bullet$ We prove that for multiple distributed sources and separable functions it suffices to protect each individual source separately, and is also necessary for sum and product functions.
	
	The paper is organized as follows. Section~\ref{settings} presents the problem formulation;  Section~\ref{main_results} summarizes our results; Proofs for these results are outlined in Sections~\ref{results_single_source} and~\ref{sec:results_multiple_sources}.

	\subsection{Related Work}
	While information theoretic security has been extensively studied, most notably by Shannon~\cite{shannon1949communication} and Wyner~\cite{wyner1975wire}, study of distortion based security was started by Yamamoto~\cite{yamamoto1988rate}, where the goal is to maximize the distortion of an eavesdropper's estimate on a message. Schieler and Cuff~\cite{schieler2014rate} later showed that, in the limit of an infinite block length ($n$) code, only $\log(n)$ bit of secret keys are needed to achieve the maximum possible distortion. However, Schieler and Cuff also showed that such secrecy is rather fragile, as causal disclosure of even a single message symbol can compromise the secrecy of the entire block. 
	This issue was because the coding scheme involves infinite block length. In this paper, we take block length equal to one, which obviates the need to wait and accumulate data at the sensor. It also removes the fragility of the distortion based measure as now we do not need to code a sequence of symbols jointly, and can rather code each symbol independently with a {new} key.  We find that for the block length one codes, while a regularly spaced alphabet can be secured easily with only one bit of shared secret key, the problem becomes combinatorially hard for irregularly spaced alphabets. We also note that, unlike information theoretic security, distortion security of alphabet $\mathcal{A}$, does not imply distortion security of alphabet $\mathcal{B}$ which is a one-to-one mapping of $\mathcal{A}$. 
	A different notion of secure estimation is studied by Wiese et. al. in \cite{wiesezeroerror} where they considered zero-error secret capacity.
	{ Semantic security has also been used for function protection, and can be posed in two different versions: (i) in an information
		theoretical setting, which has been shown to be equivalent to
		``strong'' information-theoretic security \cite{bellare2012semantic,czap2015secret} -
		and thus does not enable operation with smaller amounts of key than
		traditional strong information-theoretic security; (ii) in a computational setting,  not related to the framework of this
		paper {as we do} not make computational assumptions on the power of the
		adversary.}
	
	
	
	\begin{figure}[!t]
		\centering
		\begin{tikzpicture}
		\draw [line width=0.3mm] (2,0) circle (3mm) node {$R$};
		\draw [line width=0.3mm] (-1,1) circle (3mm) node {$S_1$};
		\draw [line width=0.3mm] (-1,0.2) circle (3mm) node {$S_2$};
		\draw [line width=0.3mm] (-1,-0.3) node {$\vdots$};
		\draw [line width=0.3mm] (-1,-1.1) circle (3mm) node {$S_n$};
		\draw [line width=0.3mm] (-0.7,1)--(1.7,0);
		\draw [line width=0.3mm] (-0.7,0.2)--(1.7,0);
		\draw [line width=0.3mm] (-0.7,-1.1)--(1.7,0);
		\end{tikzpicture}
		\caption{The sources are connected to the receiver through  noiseless channels. The receiver computes $f(X_1,  \dotsc, X_n)$.}
		\label{fig:settings}			
	\end{figure}
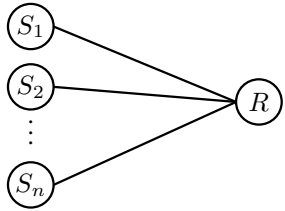
	
	
	\section{Setup}
	
	\label{settings}
	Our system consists of $n$ sources $S_1, S_2, \ldots, S_n$ connected through noiseless links to a common receiver as shown in Fig.~\ref{fig:settings}. An eavesdropper, Eve, observes all the $n$  links. 
	Each source produces a symbol $X_i$ drawn with a distribution $P_i$ from a discrete alphabet\footnote{For many systems, measurements are over reals, and we discretize by quantization at each source.} set $\mathcal{X}_i$ of size $m_i$. We assume that the symbols generated at  the $n$ sources are mutually independent. 
	
    The receiver wants to compute a function $f$ (with co-domain $\mathbb{R}$) of source symbols $\underbar{X}:= [X_1,X_2,\ldots, X_n]$, such that the receiver's computation has no error and Eve's estimate is ``maximally distorted'' with respect to the function $f$. 
    Each source is also given a secret key of size $k_i$ bits shared with {only the receiver and not with the eavesdropper}; equivalently a symbol ($K_i$) uniformly drawn from  $\{1,2,\ldots, 2^{k_i}\}$. 
    {We denote the symbol transmitted by source $S_i$ as {$g(X_i,K_i)$} and shorten it as $g(X_i)$ for brevity.
    With $g(\underbar{X},\underbar{k})$ (and shorten it as $g(\underbar{X})$), we denote set of symbols transmitted from all the sources i.e., $\{g(X_i): i \in [n] \},$ where $[n]:= \{1,2,\ldots, n\}$. 
	The distortion of Eve's estimate {for} $\underbar{k}:= [k_1, k_2, \ldots, k_n]$ bits of keys is:}
	\begin{equation}
	D_{ach}(\underbar{k})  = \min_{\hat{f}}\mathbb{E}_{\underbar{X}, \underbar{k}} [(f(\underbar{X})-\hat{f}(g(\underbar{X}, \underbar{k})))^2], 
	\end{equation}
	where $\hat{f}$ is the Eve's {strategy to estimate $f$ from her observation $g(\underbar{X},\underbar{k})$}.
	If { Eve has no information}, { her distortion is the maximum distortion and is denoted as $D_{max}$. In this case, Eve's best estimate is $\mathbb{E}[f]$ and the corresponding distortion is}
	\begin{equation}
	D_{max} = \text{var}(f(\underbar{X})).
	\end{equation}
	We are interested in the following two problems.\\
	$\bullet$ For a given function $f$, what is the minimum amount of shared keys required so that $D_{ach}(\underbar{k})$ is equal to $D_{max}$.\\
	$\bullet$ For a given function $f$ and amount of keys $\underbar{k}$, 
	what  distortion $D_{ach}(\underbar{k})$ an optimal encoding scheme achieves.
	
	
	\section{Main Results}
	\label{main_results}
	The following two theorems summarize our main results. 
	The first one deals with single-source computation, and the second one deals with multi-source computation, both under the assumption of uniform source alphabet. Proofs are outlined in sections~\ref{results_single_source} and~\ref{sec:results_multiple_sources}, respectively.
	
	\begin{thm} \label{thm:single_decay}
		Suppose there is a single source modeled as a random variable $X$ connected to the receiver through a noiseless channel, and the receiver wishes to compute a function $Y = f(X)$ with perfect accuracy. 
		Assume $k$ bits of secret key are shared between the transmitter and the receiver.
		If  $Y$ is uniformly distributed on some alphabet $\mathcal{Y} \subseteq \mathbb{R}, \ |\mathcal{Y}| = m$, then the difference between the maximum distortion of $Y$ when no information is available and the eavesdropper's achievable distortion can be upper bounded by\footnote{Results do not depend on the value $m$ takes and are true $\forall m \in \mathbb{N}$.}
		\begin{equation}\label{eq1}
		\Delta  = D_{max} - D_{ach}(k) \leq  \frac{D_{max}}{2^k}.
		\end{equation}
		Furthermore, if $d = \max\limits_{y_i,y_j \in \mathcal{Y}} (y_i - y_j) $, then
		\begin{equation}\label{eq2}
		\Delta  = D_{max} - D_{ach}(k) \leq  \frac{d^2}{2^{2k}}.
		\end{equation}
	\end{thm}
	

	Theorem~\ref{thm:single_decay} implies that each bit of the shared key exponentially increases the eavesdropper's distortion toward $D_{max}$. In particular, 
	with just $5$ bits of shared key we can achieve a distortion close to $97\%$ of $D_{max}$. In contrast, for perfect information theoretic secrecy, we require $\log(m)$ bits of key.
	
	The proof of Theorem~\ref{thm:single_decay} is constructive. Depending on the key value $i$, the source uses a different mapping $\sigma_i$ to create the symbols to transmit.  Algorithms~\ref{greedy_binning} and~\ref{exponential_binning} describe how to create the mappings to achieve~(\ref{eq1}) and~(\ref{eq2}), respectively.  Alg.~\ref{greedy_binning} has polynomial-time complexity, is optimal for $k=1$, and achieves perfect secrecy when $\mathcal{Y}$ is ``regularly spaced'' (see Section~\ref{results_single_source}).
	\begin{thm} \label{thm:separable}
		If $f(\underbar{X}) = \sum\limits_{\ell=1}^{L} \prod\limits_{i=1}^{n}f_i^{(\ell)}(X_i)$,  where $f_i^{(\ell)}(X_i)$ is an arbitrary function for all $i \in [n], \ell \in [L] $ then it is sufficient to secure each individual $f_i^{(\ell)}(X_i)$.
		Furthermore, if
		{
			\begin{align*}
			1) f(\underbar{X}) & = \sum\limits_{i=1}^n f_i(X_i), \mbox { or } \\
			2)  f(\underbar{X}) & = \prod_{i=1}^n f_i(X_i) \mbox { and } \ \prod_{i=1}^n \mathbb{E}[f_i(X_i)]. var[f_i(X_i)] \neq 0,
			\end{align*}
		}	
		it is also necessary to secure the individual functions $f_i(X_i)$.
	\end{thm}
	
	Theorem~\ref{thm:separable} states that to protect a function that can be written as a sum of product separable functions, it suffices to protect each individual component separately. This includes a fairly wide class of functions that can be protected. Theorem~\ref{thm:separable} also shows that such a scheme is optimal when we wish to compute sum and product under our setting.
	
	As an example, suppose there are $n$ source symbols $X_1, X_2, \dotsc, X_n$ all i.i.d. on $[m]$.  If  the function to compute is $\sum_{i=1}^n X_i$, then we can protect each $X_i$ individually with $1$ bit of key, and thus $n$ bits of key would achieve perfect distortion based secrecy. It can be shown that $n \log(m)$ bits will be necessary to information-theoretically secure the sum function.
	

	\begin{algorithm}
		\caption{}\label{greedy_binning}
		\begin{algorithmic}[1]
			\Statex WLOG assume $y_1 \geq y_2 \geq \ldots \geq y_m$. We use $r = m$ and select permutations $\sigma_i$ on $[m]$ for $ i \in [2^k]$. The source uses the mapping $\sigma_i$ if the shared key $K$  value is $i$. 
			\Statex \textbf{Select Permutations:}
			\State  $\sigma_1(j) = j, \ \forall j \in [m]$ \quad  \text{ (\em $\sigma_1$ is the identity permutation)}
			\For {$2 \leq i \leq 2^k$}  \quad \text{(\em determine each $\sigma_i$ at  step $i$)}
			\State $\hat{S_j} \gets \sum\limits_{\ell = 1}^{i-1} y_{\sigma_\ell(j)}, \ \forall j \in [m]$  
			\State Arrange $\hat{S_j}$ in increasing order. Let $\kappa_1, \kappa_2, \ldots, \kappa_m$ be the corresponding indices, namely $\hat{S_{\kappa_1}} \leq \hat{S_{\kappa_2}} \leq \ldots \leq \hat{S_{\kappa_n}}$ then $j$ gets permuted to $\kappa_j^{th}$ position. That is, $\sigma_i(\kappa_j) = j, \ \forall j \in [m]$.
			\EndFor{}
		\end{algorithmic}
	\end{algorithm}	
	
	\begin{algorithm}
		\caption{}\label{exponential_binning}
		\begin{algorithmic}[1]
			\State Make $2^k$ copies of the alphabet $\mathcal{Y} = \{y_1, y_2, \ldots, y_m\}$. WLOG assume $y_1 \geq y_2 \geq \ldots \geq y_m$. 
			\State Randomly divide the $2^k$ copies of $\mathcal{Y}$ into $r$ bins such that each bin has less than or equal to $2^k$ elements. Here, each bin will correspond to a symbol the $y$'s got mapped to.
			\State $S_i \gets \sum_{y \in \text{bin } i} y$ 
			\State while $\max{(S_i)} - \min(S_j) > y_1 - y_m$, exchange the largest element in bin $i$ with smallest element in bin $j$.
		\end{algorithmic}
	\end{algorithm}

	\section{Proofs for Single Source} \label{results_single_source}
	We here consider a single source $S$, a message $X$ drawn from a discrete alphabet $\mathcal{X}$, and a receiver who wants to compute a function $Y=f(X)$ that takes $m$ values (without loss of generality, assume that $y_1 \geq y_2 \geq \ldots \geq y_m$).
	The source and receiver use $k$ shared secret bits to protect $Y$.
	
	Depending on the value $K=i$ the key bits take, the source uses a different mapping $\sigma_i$ that maps each $y_j$ value to one of $r$ transmission symbols $\tau_j$.  We can depict all the $\sigma_i$ mappings together with a bipartite graph as in Fig.~\ref{fig:gen_mapping_single_source} (we allow parallel edges), where we think of every $\tau_j$ as a bin. We denote  the number of $y$-symbols  mapped into this bin  as $N_j$ and the sum of all these symbols by $S_j$, where $j \in [r] $.
	We refer to $S_j$ as the $j^{th}$ bin value.
	
	For example, the two mappings in Fig.~\ref{fig:gen_mapping_single_source}(b), depicted with solid and dotted lines, are created using Algorithm~\ref{greedy_binning}: we have four message symbols ($m =4$), four bins ($r = 4$), each bin has two elements  ($N_j = 2, \forall j$) and  bin value five   ($S_j = 5, \forall j$).

	\begin{figure}[!t]\centering
		\subcaptionbox{}{
			\begin{tikzpicture}
			\draw (1,1) node {\textbf{$Y$}};
			\draw [line width=0.3mm] (1,0) circle (4mm) node {$y_1$};
			\draw [line width=0.3mm] (1,-1) circle (4mm) node {$y_2$};
			\draw [line width=0.3mm] (1,-2) circle (4mm) node {$y_3$};
			\draw [line width=0.3mm] (1,-2.7) node {$\vdots$};
			\draw [line width=0.3mm] (1,-3.6) circle (4mm) node {$y_m$};
			
			\draw [line width=0.3mm] (1.4,0)--(2.2,-1.2);		
			\draw [line width=0.3mm, dashed] (1.4,0)--(2.2,-0.8);
			\draw [line width=0.3mm, dotted] (1.4,0)--(2.2,-0.4);
			\draw [line width=0.3mm, densely dotted ] (1.4,0)--(2.2,-0.1);
			
			\draw [line width=0.3mm, ->] (2,0.1) node[above] {$2^k$} .. controls (1.8, -0.4) .. (1.3,-0.5);
			
			\draw (3,1) node {\textbf{$\tau$}};
			\draw [line width=0.3mm] (3,0) circle (4mm) node {$\tau_1$};
			\draw [line width=0.3mm] (3,-1) circle (4mm) node {$\tau_2$};
			\draw [line width=0.3mm] (3,-2) circle (4mm) node {$\tau_3$};
			\draw [line width=0.3mm] (3,-2.7) node {$\vdots$};
			\draw [line width=0.3mm] (3,-3.6) circle (4mm) node {$\tau_r$};
			
			\draw [line width=0.3mm] (1.8,-2.2)--(2.6,-2);		
			\draw [line width=0.3mm, dashed] (1.8,-2.6)--(2.6,-2);
			\draw [line width=0.3mm, dotted] (1.8,-3)--(2.6,-2);
			\draw [line width=0.3mm, ->] (2,-2) node[above] {$\leq 2^k$} .. controls (2.2, -2.4) .. (2.7,-2.5);

			\end{tikzpicture}}\quad
		\subcaptionbox{}{
			\begin{tikzpicture}
			\draw (1,1) node {\textbf{$Y$}};
			\draw [line width=0.3mm] (1,0) circle (4mm) node {$1$};
			\draw [line width=0.3mm] (1,-1) circle (4mm) node {$2$};
			\draw [line width=0.3mm] (1,-2) circle (4mm) node {$3$};
			\draw [line width=0.3mm] (1,-3) circle (4mm) node {$4$};
			
			\draw (3,1) node {\textbf{$\tau$}};
			\draw [line width=0.3mm] (3,0) circle (4mm) node {$\tau_1$};
			\draw [line width=0.3mm] (3,-1) circle (4mm) node {$\tau_2$};
			\draw [line width=0.3mm] (3,-2) circle (4mm) node {$\tau_3$};
			\draw [line width=0.3mm] (3,-3) circle (4mm) node {$\tau_4$};

			\draw [line width=0.3mm, dotted] (1.4,0)--(2.6,-3);
			\draw [line width=0.3mm] (1.4,0)--(2.6,0);
			\draw [line width=0.3mm] (1.4,-1)--(2.6,-1);
			\draw [line width=0.3mm, dotted] (1.4,-1)--(2.6,-2);
			\draw [line width=0.3mm, dotted] (1.4,-2)--(2.6,-1);
			\draw [line width=0.3mm] (1.4,-2)--(2.6,-2);
			\draw [line width=0.3mm] (1.4,-3)--(2.6,-3);
			\draw [line width=0.3mm, dotted] (1.4,-3)--(2.6,0);
			
			\draw [line width=0.3mm, dotted] (1,-3.7)--(2.1,-3.7) node[right] {$K = 0$};
			\draw [line width=0.3mm] (1,-4)--(2.1,-4) node[right] {$K = 1$};		
			
			
			%
			%
			%
			\end{tikzpicture}}
		\caption{(a) Each of the $2^k$ mappings corresponds to $m$ edges in the bipartite graph; the key value determines which one to use. (b) Example of  mappings created with Algorithm~\ref{greedy_binning}.}
		\label{fig:gen_mapping_single_source}
	\end{figure}
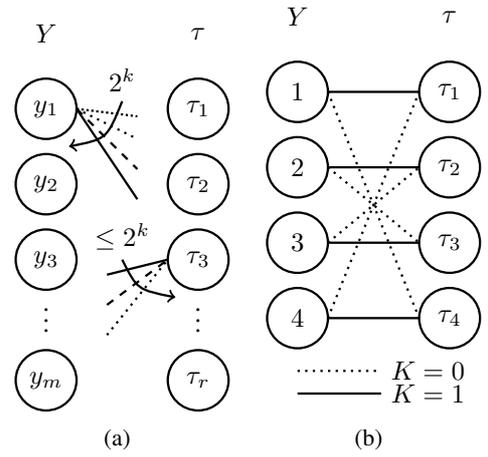
	
	We prove the  first bound in Theorem~\ref{thm:single_decay} by showing that the mappings created by Algorithm~\ref{greedy_binning} achieve at least the distortion in (\ref{eq1}).
	Lemma~\ref{lem:diff_expression} shows that maximizing distortion is equivalent to  minimizing
	the sum of the bin values $S_j$. Given that $\sum_{j=1}^r S_j$=$2^k\sum_{j=1}^n y_j$=constant, the minimum  is achieved when all the bin values are equal.
	Alg.~\ref{greedy_binning} tries to achieve this for the case where $r=m$ with a greedy approach that sequentially designs each of the mappings $\sigma_i$, $i=1,2,\ldots, 2^k$: at step~$i$, it
	calculates the partial bin values $\hat{S_j}$ that sum $y$-values allocated in the previous $i-1$ mappings; for the $\sigma_i$ mapping, it allocates the smaller symbol values to the bins that have accummulated the larger  $\hat{S_j}$ values so as to smooth out differences.
	
	
	
	\begin{restatable}{lem}{lemDiffExpression} \label{lem:diff_expression}
		
		\begin{align*}
		\Delta & = D_{max} - D_{ach}(k) \\
		&= \sum_{j=1}^{m}\! \frac{1}{p(\tau_j)} \left(\sum_{i=1}^{m}\!y_i p(\tau_j|y_i)p(y_i)\right)^2 \!\!- \! \mathbb{E}[Y]^2.  
		\end{align*}  
		In addition, if the source alphabet is uniformly distributed, then the expression simplifies to
		\begin{align*}
		\Delta & = D_{max} - D_{ach}(k) = \frac{1}{2^{k} m}\sum_{j=1}^{r} \frac{S_j^2}{ N_j} - \mathbb{E}[Y]^2.  	\end{align*}  
		where $S_j$ is the sum of all symbols  (not necessarily distinct) mapped to $\tau_j$.
	\end{restatable}
	{\em Proof.}
	If Eve observes a symbol $\tau$, then her best estimate is $\mathbb{E}\left[ Y \big | \tau  \right]$, with distortion $D_{ach} = \mathbb{E}\left[ \text{var} (Y) \big | \tau  \right] $. Thus 
	\begin{align*}
	&\Delta  = \text{var}(Y) - \mathbb{E}\left[ \text{var} (Y) ) \big | \tau  \right] 
	= \text{var}\left(\mathbb{E} \left[ Y | \tau \right] \right) \\
	&= \sum_{j=1}^{r} p(\tau_j) \left(\sum_{i=1}^{m} y_i p(y_i | \tau_j)\right)^2 - \mathbb{E}\left[Y \right]^2 \\
	&= \sum_{j=1}^{r} \frac{1}{p(\tau_j)} \left(\sum_{i=1}^{m} y_i p(\tau_j|y_i)p(y_i)\right)^2 - \mathbb{E}\left[Y \right]^2.
	\end{align*}
	
	For uniform $Y$, let $n_{ij}$ be the number of keys for which $y_i$ is mapped to $\tau_j$. Thus the difference becomes 
	\begin{align*}
	\Delta & = \sum_{j=1}^{r} \frac{1}{p(\tau_j)} \left(\sum_{i=1}^{m} y_i p(\tau_j|y_i)p(y_i)\right)^2 - \mathbb{E}\left[Y \right]^2 \\ 
	&= \sum_{j=1}^{r} \frac{1}{N_j/\left(2^k m\right)} \left(\sum_{i=1}^{m} y_i \left(\frac{n_{ij}}{2^k} \right) \left(\frac{1}{m}\right)\right)^2 - \mathbb{E}\left[Y \right]^2 \\
	&= \frac{1}{2^k m} \sum_{j=1}^{r} \frac{1}{N_j} \left(\sum_{i=1}^{m} y_i n_{ij}\right)^2 - \mathbb{E}\left[Y \right]^2 \\
	&= \frac{1}{2^{k} m}\sum_{j=1}^{r} \frac{1}{ N_j} \left(S_j\right)^2 - \mathbb{E}\left[Y \right]^2.
	\end{align*}
	\begin{proof}[Proof of~(\ref{eq1}) in Theorem~\ref{thm:single_decay}]
		For $1 \leq i \leq 2^k$, we use the $\sigma_i$ from Alg.~\ref{greedy_binning}, and map $\{y_{\sigma_i(j)} : i \in [2^k]\}$ for $j \in [m]$ to bin $j$. Using Lemma~\ref{lem:diff_expression},   
		\begin{align*}
		\Delta & =
		D_{max} - D_{ach}(k) =
		\frac{1}{2^{k} m}\sum_{j=1}^{r} \frac{S_j^2}{N_j}  - \mathbb{E}[Y]^2 \\
		& \stackrel{(i)}= \frac{1}{2^{2k} m}\sum_{j=1}^{m} \left( \sum\limits_{i=1}^{2^k} y_{\sigma_i (j)}\right)^2  - \mathbb{E}[Y]^2 \\
		& = \frac{1}{2^{2k} m} \left(\!\!\sum_{i=1}^{2^k}  \sum\limits_{j=1}^{m}\!\!y^2_{\sigma_i (j)}\!\! +\!\! 2 \sum_{\ell=2}^{2^k} \sum\limits_{j=1}^{m}\!\!y_{\sigma_\ell (j)} \sum\limits_{i=1}^{\ell-1}\!\!y_{\sigma_i (j)}\!\! \right)\!-\!\mathbb{E}[Y]^2  \\
		& \stackrel{(ii)}{\leq} \frac{1}{2^{2k} m} 2^k  \sum\limits_{j=1}^{m} y^2_j \\
		& \quad + \frac{2}{2^{2k} m} \sum_{\ell=2}^{2^k} m \left( \frac{\sum\limits_{j=1}^{m} y_j}{m} \right) \left( \frac{\sum\limits_{j=1}^{m} \sum\limits_{i=1}^{\ell-1} y_{\sigma_i (j)} }{m}\right)  - \mathbb{E}[Y]^2 \\
		& = \frac{1}{2^{k}}   \mathbb{E}[Y^2] + \frac{2}{2^{2k} m} \sum_{\ell=2}^{2^k} m \mathbb{E}[Y]  (\ell-1) \mathbb{E}[Y]   - \mathbb{E}[Y]^2 \\
		& = \frac{1}{2^{k}}   \mathbb{E}[Y^2] + \frac{2 \mathbb{E}[Y]^2}{2^{2k}} \frac{(2^k-1)2^k}{2}  - \mathbb{E}[Y]^2 \\
		& = \frac{1}{2^{k}}   \mathbb{E}[Y^2] - \frac{1}{2^{k}} \mathbb{E}[Y]^2 = \frac{1}{2^{k}} \text{var}(Y) =  \frac{1}{2^{k}} D_{max}.
		\end{align*} 
		
		Here $(i)$ follows by putting $N_j = 2^k$ and $(ii)$ follows from Chebyshev's sum inequality~\cite{hardy1988inequalities} as $y_{\sigma_\ell (j)}$ and $\sum\limits_{i=1}^{\ell-1} y_{\sigma_i (j)}$ are in opposing order due to the construction in Alg.~\ref{greedy_binning}.		
	\end{proof}
	To prove the second bound in Theorem~\ref{thm:single_decay}, we use Alg.~\ref{exponential_binning} to design the $\sigma_i$'s;
	Lemma~\ref{lem:bin_sum} bounds the resulting $S_j$ values.
	\begin{restatable}{lem}{lemBinSum} \label{lem:bin_sum}
		Let $d = \max{(y_i)} - \min{(y_j)}$, then at the end of  Alg.~\ref{exponential_binning}, $\mathbb{E}[Y] 2^k \frac{m}{r} - d \leq S_i \leq \mathbb{E}[Y] 2^k \frac{m}{r} + d, \ \forall i \in [r]$.
	\end{restatable}
	\begin{proof}
		Suppose the contrary that there is some $i$ such that $S_i > \mathbb{E}[Y] 2^k \frac{m}{r} + d$. 	By the pigeonhole principle, there is also some $j$ so that $S_j < \mathbb{E}[Y] 2^k \frac{m}{r}$. But then $S_j - S_i > d$, which contradicts the terminating condition of the algorithm.
		Similar reasoning holds for the case where there is some $i$ such that $S_i < \mathbb{E}[Y] 2^k \frac{m}{r} - d$.
	\end{proof}

	\begin{proof}[Proof for (\ref{eq2}) in Theorem \ref{thm:single_decay}]
		We first note that, in each step of  Alg.~\ref{exponential_binning}, $\max{(S_i)} - \min(S_j)$ is  decreasing. Since there are only finite many possible values for $\max{(S_i)} - \min(S_j)$, the algorithm terminates. 
		
		Take $r = m$, so there are $2^k$ elements in each bin, and so $\Delta = \frac{1}{2^{2k} m} \sum\limits_{i=1}^m S_i^2 -  \mathbb{E}[Y]^2$. By Lemma~\ref{lem:bin_sum}, we have $ \mathbb{E}[Y] 2^k -d \leq S_i \leq  \mathbb{E}[Y] 2^k + d$. We also note that $\sum\limits_{i=1}^{m} S_i = m 2^k  \mathbb{E}[Y]$. So to bound $\Delta$, we have the following optimization problem:
		\begin{align*}
		\begin{array}{ll}
		{\rm{maximize}} & \sum\limits_{i=1}^m S_i^2 
		\\ {\rm{subject \ to}} &  \sum\limits_{i=1}^{m} S_i   =  \ m 2^k \mathbb{E}[Y]
		\\  & \mathbb{E}[Y] 2^k - d \  \leq S_i   \leq \ \mathbb{E}[Y] 2^k + d, \ \forall i \in [m].
		\end{array}
		\end{align*}	
		 Then $ \sum\limits_{i=1}^m S_i^2 $ is upper bounded by $m (2^{2k} \mathbb{E}[Y]^2\!+\!d^2)$. Thus,
		\begin{align*}
		\Delta & \leq \frac{1}{2^{2k}} (2^{2k} \mathbb{E}[Y]^2 + d^2) -  \mathbb{E}[Y]^2  = \frac{d^2}{2^{2k}}.
		\end{align*} 
	\end{proof}
	
	We next characterize properties of optimal mappings, using the bipartite graph representation   in Fig.~\ref{fig:gen_mapping_single_source}.
	The following lemma bounds the degrees and number of vertices in the graph, which we use to prove Theorem~\ref{thm:optimality}.
	\begin{restatable}{lem}{lemGraph} \label{lem:graph}
		For the mapping in Fig.~\ref{fig:gen_mapping_single_source}, the following are true.
		\begin{enumerate}		
			\item The degree of  each $y_j$ is $2^k$.
			\item The degree of each $\tau_j$ is $\leq 2^k$. This implies $N_j \leq 2^k$.
			\item For an optimal mapping, $N_j \leq 2^{k-1}$ for at most one $j$.
			\item For an optimal mapping, $m\leq r < 2m$.
		\end{enumerate}
	\end{restatable}
	
	\begin{proof}
		\textcolor{white}{.}1)  Since the key $K$ can take $2^k$ different values, there will be one outgoing edge for each value of key.\\ 
		2) The legitimate receiver needs to decode $Y$, and a degree greater than $2^k$ will make it impossible to decode based on the shared key.\\
		3) Suppose there are two bins having less than or equal to $2^{k-1}$ elements, then we can merge them. Eve's distortion before merging is
		$D = \min_{\hat{Y} = g(\tau)}\mathbb{E} [(Y-\hat{Y})^2].$
		Now, call the random variable representing the transmitted symbol after merging  $\tau'$. This gives distortion
		$D = \min_{\hat{Y} = g(\tau')}\mathbb{E} [(Y-\hat{Y})^2]$.
		Since $\tau'$ is a function of $\tau$, this gives a higher distortion.\\
		4) The net degree of $Y$ is $m2^k$. Similarly, the net degree of $\tau$ is not greater than $r 2^k$. By part 1 of the lemma, $m2^k \leq r 2^k$. This implies $r \geq m$. Based on the third part of this lemma, each bin will have more than $2^{k-1}$ edges except probably one. This bounds the total degree ($\mathcal{I}$) of $\tau$ as $m 2^k  = \mathcal{I} \geq (r-1)(2^{k-1} + 1) + 1 >   (r-1)(2^{k-1}).$ Thus $r \leq 2m$. 
	\end{proof}
	
	\begin{thm}
		\label{thm:optimality}
		The encoding scheme in Alg.~\ref{greedy_binning} is optimal for $k=1$. Moreover, if $\mathcal{Y}$ is ``regularly spaced'' i.e. $\mathcal{Y} = \{y, y + d,\ldots, y+(m-1)d \}, \text{ for some } y, d \in \mathbb{R}, d \leq 0$, it  achieves perfect distortion security.  
	\end{thm}
	
	\begin{proof}
		By Lemma~\ref{lem:graph}, $N_j = 2$ for each $j$, so
		\begin{align*}
		& \Delta  = D_{max} - D_{ach}(k) =  \frac{1}{2^{2k} m}\sum_{j=1}^{m} S_j^2 - \mathbb{E} [Y]^2 \\
		\quad & = \frac{1}{4 m}\sum_{j=1}^{m}\!\! S_j^2\!-\!\mathbb{E} [Y]^2 = \frac{1}{4m} \sum_{j=1}^m (y_{\sigma_1(j)} + y_{\sigma_2(j)})^2 - \mathbb{E}[Y]^2 \\
		\quad & = \frac{1}{2m}\left(\sum_{j=1}^m y_j^2 + \sum_{j=1}^m y_{\sigma_1(j)} y_{\sigma_2(j)} \right) - \mathbb{E}[Y]^2. 
		\end{align*}
		
		Without loss of generality, assume that $\sigma_1$ is the identity function, then by the Rearrangement Inequality, $\Delta$ is minimized when $\sigma_2(j) = m-1-j$. This is the same $\sigma_2(j)$ we get from Alg.~\ref{greedy_binning}.	Now for the regularly spaced $\mathcal{Y}$,
		\begin{align*}
		S_j & = y_j + y_{m-j+1} = 2y + (m-1)d, \\
		\mathbb{E} [Y | \tau_j] & = \frac{S_j}{2} = y + \frac{(m-1)}{2} d 
		=\mathbb{E} [Y], \ \forall j \in [m]. 
		\end{align*}
		
		Thus by Lemma~\ref{thm:max_dist_cond}, $D_{ach}(1)  = D_{max}$.
	\end{proof}
	
	\section{Proofs for Multiple Sources} \label{sec:results_multiple_sources}
	To prove Theorem~\ref{thm:separable}, we repeatedly apply  Lemma~\ref{thm:max_dist_cond}, which follows from 
	standard results in MMSE estimation  \cite{kailath2000linear}. 
	
	\begin{lem}
		\label{thm:max_dist_cond}
		The function $f(\underbar{X})$ is perfectly secured under the Euclidean distortion measure if and only if, 
		\begin{align*}
		\mathbb{E} \left[ f(\underbar{X}) \big | g(\underbar{X})\right] = \mathbb{E} \left[ f(\underbar{X})\right], \ \forall g(\underbar{X}). 
		\end{align*}
	\end{lem}

	\begin{proof}[Proof of Theorem \ref{thm:separable}]
		{\bf Sufficient:} If $f_i^{(\ell)}(X_i)$'s are secure, then 
		\begin{align*}
		& \mathbb{E} \left[ f(\underbar{X}) \big | g(\underbar{X})\right]	
		=  \mathbb{E} \left[\sum\limits_{\ell=1}^{L} \prod_{i=1}^{n}f_i^{(\ell)}(X_i)| g(\underbar{X}) \right] \\
		\quad & = \sum\limits_{\ell=1}^{L}  \prod_{i=1}^n \mathbb{E} \left[ f_i^{(\ell)}(X_i) | g(\underbar{X}) \right] = \sum\limits_{\ell=1}^{L}  \prod_{i=1}^n \mathbb{E} \left[ f_i^{(\ell)}(X_i)  \right] \\
		\quad & = \mathbb{E} \left[ \sum\limits_{\ell=1}^{L}   \prod_{i \in [n]}  f_i^{(\ell)}(X_i) \right] = \mathbb{E} \left[ f(\underline{X}) \right].
		\end{align*}
		
		
		{\bf Necessary:} Suppose not all $f_i(X_i)$ are secured. For $f(\underbar{X})  = \sum_{i=1}^n f_i(X_i)$, by Lemma~\ref{thm:max_dist_cond} we can pick $g(X_i)$ so that $\mathbb{E}\left[ f_i(X_i) |g(X_i) \right] \neq \mathbb{E}\left[ f_i(X_i) \right]$ for each unsecured $f_i(X_i)$. Since $\mathbb{E}\left[\mathbb{E}\left[ f_i(X_i) |g(X_i) \right] \right] = \mathbb{E}[f_i(X_i)] $ we can further assume without loss of generality that $\mathbb{E}\left[ f_i(X_i) |g(X_i) \right] > \mathbb{E}\left[ f_i(X_i) \right]$. With this,
		\begin{align*}
		& \mathbb{E}\left[ f(\underbar{X}) \big | g(\underbar{X}) \right] 
		= \sum\limits_{i=1}^n \mathbb{E}\left[ f_i(X_i) \big | g(\underbar{X}) \right] \\
		\quad & > \sum\limits_{i=1}^n \mathbb{E}\left[ f_i(X_i) \right]  = \mathbb{E} \left[ f(\underbar{X}) \right]. 
		\end{align*}
		Thus $ f(\underbar{X})  = \sum\limits_{i=1}^n f_i(X_i)$ is not secure.
		
		Similarly for the case $f(\underbar{X})  = \prod_{i=1}^n f_i(X_i)$, for each unsecured $f_i(X_i)$, we pick $g(X_i)$ so that $\lvert\mathbb{E}\left[f_i(X_i) | g(X_i)\right]\rvert > \lvert\mathbb{E}\left[f_i(X_i)\right]\rvert$. Then we have
		\begin{align*}
		& \big\lvert\mathbb{E}\left[ f(\underbar{X}) \big | g(\underbar{X}) \right]\big\rvert  = \prod\limits_{i=1}^n \big\lvert\mathbb{E}\left[ f_i(X_i) \big | g(\underbar{X}) \right]\big\rvert \\
		& > \prod\limits_{i \in [n]} \big\lvert\mathbb{E}\left[ f_i(X_i) \right]\big\rvert  = \big\lvert\mathbb{E}\left[f(\underbar{X})\right]\big\rvert.
		\end{align*} 
		Thus  $f(\underbar{X})  = \prod_{i=1}^n f_i(X_i)$ is also not secure. It remains to show it is necessary to transmit each $f_i(X_i)$.
		
		Suppose the receiver uses a function $h$ on $\{g(X_i), i \in [n]\} $ to compute $f(\underbar{X})$. If $f(\underbar{X})  = \sum\limits_{i=1}^n f_i(X_i)$, then 
		\begin{align*}
		f_i(X_i) & = h\left(g(x_1),\ldots, g(X_i), \ldots, g(x_m)\right) -  \sum\limits_{j \neq i} f_j(x). 	
		\end{align*}
		If $f(\underbar{X})  = \prod_{i=1}^n f_i(X_i) $,
		\begin{align*}
		f_i(X_i) & = \frac {h\left(g(x_1),\ldots, g(x_{i-1}), g(X_i), g(x_{i+1}), \ldots, g(x_n)\right)} {\prod\limits_{j \neq i} f_j(x_j)}, 
		\end{align*}
		where we choose $x_j$, such that $f_j(x_j) \neq 0$. 
		Thus $f_i(X_i)$ are necessary to communicate with the receiver.
	\end{proof}
	
	\section{Conclusions}
	\label{discussions}
	In this paper, we argued that distortion based security is important for applications such as CPS, and presented the first coding schemes that achieve short-block length (single-shot) distortion security.    We found that  Eve's distortion increases exponentially with the number of bits of shared key, and proved that our schemes  are optimal in some cases. We provide necessary and sufficient conditions for security for a number of interesting cases,  such as for functions that can be written as sum or product of functions in individual variables. This includes various statistical functions like mean and variance.
	
	\bibliographystyle{IEEEtran}
	\bibliography{ISIT2017_distortion}

\end{document}